\documentclass[authoryear,1p, 11pt]{elsarticle}

\usepackage{fancybox,graphicx,dsfont,pgf}
\usepackage{amsmath,amsthm,amsfonts,amssymb,fancybox,float}
\usepackage{subfigure,color,wrapfig}

\usepackage{soul}
 \floatstyle{ruled}
\restylefloat{table}
  \floatstyle{ruled}
\numberwithin{equation}{section}

\newcommand{\mX}{\mathbf{X}}

\newcommand{\mB}{\mathbf{B}}
\newcommand{\mC}{\mathbf{C}}

\newcommand{\mV}{\mathbf{V}}

\newcommand{\mQ}{\mathbf{Q}}

\def\RR{{\mathbb R}}

\newcommand{\la}{\lambda}
\newcommand{\La}{\Lambda}

\newcommand{\E}{\mathds{E}}
\renewcommand{\Pr}{\mathds{P}}

\newtheorem{theorem}{Theorem}

\newtheorem{lemma}{Lemma}

\theoremstyle{definition}

\theoremstyle{remark}
\newtheorem{remark}{Remark}[section]

\usepackage{amssymb}

\journal{Theoretical Population Biology}

\begin{document}

\begin{frontmatter}

\title{Separation of the largest eigenvalues in eigenanalysis of  genotype data from discrete subpopulations}

\author[label1]{Katarzyna Bryc\corref{cor1}}
\ead{kbryc@genetics.med.harvard.edu}
\author[label2]{Wlodek  Bryc}
\ead{Wlodzimierz.Bryc@uc.edu}
\author[label3]{Jack W. Silverstein}
\ead{jack@ncsu.edu}
\address[label1]{Department of Genetics, Harvard Medical School, Boston, MA 02115, USA}
\address[label2]{Department of Mathematical Sciences, University of Cincinnati, PO Box 210025, Cincinnati, OH 45221--0025, USA}
\address[label3]{Department of Mathematics, Box 8205, North Carolina State University, Raleigh, NC 27695-8205, USA}
\cortext[cor1]{Corresponding author, \textit{Phone:} 617-432-1101}

\begin{abstract}
We present a mathematical model, and the corresponding mathematical analysis, that justifies and quantifies the use of principal component analysis of biallelic genetic marker data for a set of individuals to detect the number of subpopulations represented in the data.  We indicate that the power of the technique relies more on the number of individuals genotyped than on the number of  markers. %
\end{abstract}


\end{frontmatter}

\section{Introduction}

Principal component analysis (PCA) has been a powerful and efficient method for analyzing large datasets in population genetics since its early applications by Cavalli-Sforza and others \citep{menozzi1978synthetic, cavalli1993demic, cavalli1994history}. In particular, PCA of single nucleotide polymorphism (SNP) genotype data can be used to illuminate population structure \citep{Nelson2008},  provide insights into human history and admixture \citep{novembre2008genes, McVean:2009}, and help to estimate the number of distinct subpopulations within a sample \citep{Patterson:2006}.

A good estimate for the number of subpopulations, $K$, is needed in Bayesian clustering algorithms such as \textit{STRUCTURE} \citep{Falush2003} or \textit{ADMIXTURE} \citep{Alexander2009}, where one must specify \textit{a priori} the number of clusters in the data, which affects the inferred relationships among individuals \citep{waples2006invited,latch2006relative}. \textcolor{black}{Likewise, the number of subpopulations  is informative of} how many principal components are capturing \textcolor{black}{meaningful} substructure within the data, rather than stochastic noise. Both clustering methods and PCA have been applied to learn about populations from a wide variety of species, including humans~\citep{Rosenberg2002}, \textcolor{black}{chickens~\citep{Rosenberg2001}}, cows~\citep{Gibbs2009}, canines~\citep{Ostrander2005}, and arabidopsis~\citep{Pico2008}.

In this paper, we provide additional mathematical confirmation for the use of PCA in estimating the number of subpopulations within a sample. In a related result \cite[Theorem 3]{Patterson:2006} that motivated this research, the authors analyze the theoretical centered covariance matrix for a single marker as the number of individuals increases without bound.  Here we analyze a mathematically more complicated object: the sample covariance matrix based on multiple markers. In current practice, the sample covariance matrix is often centered, and the data rows are often further  normalized.  In contrast to previous work, our results describe behavior of the eigenvalues of the sample covariance matrix \textit{without} centering or normalization, taking into account both the number of individuals and the number of markers. The raw unprocessed covariance matrix is more amenable to mathematical analysis, and the singular values of such raw data exhibit quantifiable properties that can be used directly to determine the number of subpopulations in the data in an almost deterministic fashion, at least when the number of individuals in the study is sufficiently large.

We show that for large data sets of individuals from $K$ well-differentiated subpopulations, with overwhelming probability the un-centered sample covariance matrix has  $K$ large eigenvalues.  (The technical meaning of ``well differentiated subpopulations" is that matrix $\mQ$, which we later define in equation \eqref{QQQ} from the pairwise moments of pairwise  site spectra, is non-singular.) These \textcolor{black}{``large''} eigenvalues, which indicate the presence of population structure, are greater by a factor proportional to the number of individuals than the remaining smaller eigenvalues. The large eigenvalues arise from the mixed moments of the pairwise site frequency spectra \textcolor{black}{stemming from the presence of multiple subpopulations}, while the small eigenvalues are attributed to random differences between the individuals in the sample. In practice, with finite populations, we can detect only the eigenvalues that are well separated from zero where the cutoff described in equation \eqref{Threshold} is beyond the boundary of the range of the many smaller eigenvalues, which we will refer to as the ``bulk'' of the eigenvalues.

We note that the {\em eigenvectors} of a sample covariance matrix are also interesting, but notoriously difficult to analyze mathematically, so this paper is devoted solely to understanding the eigenvalues. We believe that our model makes minimal assumptions about the distribution of the entries of the data matrix, and should work well in practice for genotype data.

\section{Methods}
 \label{Sec:patterson}

In setting up the mathematical model, we begin as in \citet{Patterson:2006}. We consider unrelated diploid individuals with independent biallelic markers. We assume that  the data for our biallelic markers are recorded in a large $M\times N$ rectangular array $\mC$ with rows labeled by individuals and columns labeled by polymorphic markers.  The entries $C_{i,j}$ are the number of variant alleles for marker $j$, individual $i$, that take values $0,1$ or $2$.
We assume that we have data for $M$ individuals from $K$ subpopulations, and that we have $M_r$ individuals from the subpopulation labeled $r$ so that $M=M_1+M_2+\dots+M_K$. Often, neither $K$ nor $M_1,\dots,M_K$ are known, so we may wish to estimate the value of $K$, the number of subpopulations in the data.
If the population sampling information were known, namely, that individual $i$ is from subpopulation $r$,  the genotype probabilities for marker $j$, $\Pr(C_{i,j}=0,1,2)$ would be given by the expected allele frequencies in subpopulation $r$, as follows: %
\begin{eqnarray}\label{Hardy-Weinberg}
\Pr(C_{i,j}=0)&=&\left(1-p_r(j)\right)^2+F_{r,j}\, p_r(j)(1-p_r(j)),\\
 \Pr(C_{i,j}=1)&=&2p_r(j)(1-p_r(j))(1-F_{r,j}),\\
 \Pr(C_{i,j}=2)&= &p_r(j)^2+ F_{r,j} p_r(j)(1-p_r(j)) \label{Hardy-Weinberg3}.
\end{eqnarray}
where $p_r(j)$ is the allele frequency of marker $j$ in subpopulation $r$.
For additional flexibility we use an auxiliary set of population parameters $F_{r,j}$ that take values between $0$ and $1$. When $F_{r,j}=F_r$ has the same value for all  markers $j$,  then  $F_r$ is an average inbreeding coefficient of the $r$-th subpopulation and our formulas coincide with \cite[Eqn. (9)]{Wright:1943fk}. We recall that $F_{r,j}=0$ for populations in Hardy-Weinberg equilibrium.  We write \begin{equation}\label{Fcoeff}
F=\sup_{r,j}F_{r,j}
\end{equation}
for the largest value of the inbreeding parameter.

Our results describe the asymptotic behavior of the singular values of $\mC$ as $N$ increases. We rely on the following ``mathematical model" of how the parameters change as $N$ changes. (Some additional regularity assumptions appear in Section \ref{Sec:Th}.) \textcolor{black}{In general, our derivations rely on describing population parameters such that
each locus or individual is viewed as a random sample from the population of all loci
and individuals}.

\textcolor{black}{Our mathematical theory will depend on the existence of certain limits for our theory to hold.
To begin, f}or any pair of subpopulations labeled by $r,s\in\{1,\dots,K\}$, we assume that there are numbers $m_{r,s}$ such that
\begin{equation}\label{kkk}
m_{r,s}=\lim_{N\to\infty} \frac{1}{N} \sum_{j=1}^N p_r(j)p_s(j)
\end{equation}
\textcolor{black}{These numbers are moments that capture information about allele frequencies of the subpopulations. }

Next we assume that  the number of individuals, $M$, grows proportionally with $N$ so that $M_1(N),\dots,M_K(N)\to\infty$ are such that with $M(N)=M_1(N)+\dots+M_K(N)$, we have $M(N)/N\to d$ for some $d\geq 0$, and
$M_r(N)/M(N)\to c_r>0$ as $N\to\infty$.  (When applying the asymptotic theory to  finite values of $M$ and $N$, we use $d=M/N$ and $c_r=M_r/M$.)

The $K\times K$ array of  deterministic numbers $m_{r,s}$ together with the relative subpopulation sizes  $c_1,\dots,c_r$  are the hidden parameters that enter  our mathematical analysis.
They enter our analysis though
a $K\times K$ (deterministic) symmetric positive matrix  $\mQ$ with entries
\begin{equation}\label{QQQ}
[\mQ]_{r,s}=\sqrt{c_rc_s}m_{r,s},
\end{equation}
where $m_{r,s}$ are given by  \eqref{kkk}. %
We assume that the subpopulations are ``well differentiated" so that $\mQ$ is of full rank with eigenvalues  $\la_1\geq \la_2\geq \dots \la_K>0$. For additional discussion of these assumptions and their relation to the joint site frequency spectrum and other models of how allelic probabilities differ between subpopulations, see Section \ref{Sec:Th}.

In what follows we deviate from the  method of \cite{Patterson:2006}, \textcolor{black}{who center and standardize $\mC$.}  %
Instead, we analyze $\mC$ as a random perturbation of a finite-rank matrix (compare \cite{benaych2011eigenvalues}) and to preserve this mathematical structure we cannot use data-dependent column averages  to center and normalize the entries of the array. So instead
we work directly with the eigenvalues of the uncentered sample covariance matrix $\mC\mC'$.
This is a symmetric square matrix of size $M$, the number of individuals.  We are interested in the behavior of the eigenvalues of $\mC\mC'$ which we write in decreasing  order $\La_1\geq\La_2\geq \dots\geq\La_M$.

In this paper we prove the mathematical properties of an estimator for the number of subpopulations based on the magnitude of these eigenvalues. We  estimate  the number of subpopulations,  $K$,   as the number of eigenvalues larger than the threshold of
\begin{equation}\label{Threshold}
t' = \frac{1+F}{2}\left(\sqrt{M}+\sqrt{N}\right)^2=N\frac{1+F}{2}\left(1+\sqrt{M/N}\right)^2
\end{equation}
Equivalently, for the scaled matrix %
  \begin{equation}\label{X}
  \mX_N=\frac{1}{(\sqrt{M}+\sqrt{N})^2}\mC\mC'.
\end{equation}
we can use the more intuitive threshold:
\begin{equation}\label{ThresholdScaled}
t= \frac{1+F}{2}
\end{equation}
which does not depend on $M,N$. The parameter $F$ used here, as defined by equation
\eqref{Fcoeff}, takes values between $0$ and $1$. %

Hence begins our main result, that depending on the value of $F$, the threshold cutoff $t$ for determining the number of large eigenvalues corresponding to population structure,  is between
 $0.5$ and $1$. If there are $K$ subpopulations present in the data, then as $N$ and $M$ increase without bound (and are subject to certain technical conditions), then with overwhelming probability the smallest $M-K$ eigenvalues of $\mC\mC'$  are smaller than $t'$ from  \eqref{Threshold}.

Furthermore, the consecutive  largest $K$ eigenvalues are typically much larger.
The theoretical prediction for the observed largest  eigenvalues of the normalized  sample covariance matrix  \eqref{X} are
 \begin{equation}\label{La-predicted}
    \La_j\approx  \frac{4 MN \la_j }{\left(\sqrt{M}+\sqrt{N}\right)^2},
    \end{equation}
    where $\la_j$ are the eigenvalues of matrix $\mQ$ introduced in  \eqref{QQQ}.
 So in evaluating whether an eigenvalue of $\mX_N$ corresponds to population structure, we are effectively comparing a constant between $0.5$ and $1$ (depending on the value of $F$),  to a number larger than $\la_K M$. Of course, $\la_K M$ can be made arbitrarily large by increasing the number of individuals $M$, making it possible to resolve which eigenvalues correspond to population structure.
From our mathematical analysis, we show that the number of subpopulations is estimated essentially without error when  $M\la_K$ is larger than $1$. More specifically,
the   accuracy of this estimator of $K$ depends on the theoretical predicted smallest subpopulation eigenvalue:
\begin{equation}\label{eta}
L=\frac{4MN \la_K}{\left(\sqrt{M}+\sqrt{N}\right)^2}
\end{equation}
 for the rescaled matrix  $\mX_N$ from equation \eqref{X}.
This theoretical value indicates whether there is likely to be power to detect the full population substructure present in the data, if, for example, $L > 0.5$.
Indeed, our simulations confirm that our estimate of population structure works very well whenever $L$ is larger than $0.5$ (when $F=0$),  or $L$ is larger than $1$ (when $F=1$).   In practice, subpopulation parameter  $\la_K$ (the smallest eigenvalue of matrix $\mQ$ from equation~\eqref{QQQ}) and thereby $L$, is a hidden parameter that  depends on the theoretical hidden subpopulation moments and on the unknown relative proportions $c_1,c_2,\dots,c_K$ of the subpopulations present  in the data, and cannot be obtained for non-simulated datasets.

In view of this strong separation,
the eigenvalues of $\mC\mC'$, can be safely used in exploratory data analysis without need for a formal statistical test to assess significance of structure when $M$ is large enough.

A check for the appropriateness of the cutoff is provided by the histogram of the eigenvalues: the $K$ largest eigenvalues should be separated from the remaining eigenvalues, or the bulk. Under the model of clean population substructure, the remaining eigenvalues  should cluster together into a fairly solid group, as these eigenvalues correspond to random differences among individuals. Ideally, one expects the shape of the  bulk to be a single-mode semi-elliptical mass {with sharp boundaries} like the Marchenko-Pastur law \cite[Chapter 3]{bai2010spectral}.
After normalization \eqref{X}, the distribution of the bulk should be located to the left of $(1+F)/2$.
\subsection{Robust to violations of assumptions}
In our analysis we explore possible violations of our key assumptions, namely, independence among markers and stochastic independence of individuals drawn from a subpopulation.

For our mathematical derivations, we assume that each marker is independent. However, genetic markers on the same chromosome are inherited together, leading to a non-random correlation of markers, or ``linkage disequilibrium'' (LD). Simulations indicate that LD does not strongly affect our ability to detect population structure, though it does violate our assumptions. In our view, thinning the data by removing one SNP from each pair of highly correlated markers (such as via the LD-pruning implemented in \textit{PLINK}) is a simple yet robust technique that addresses linkage disequilibrium violations of independence assumptions, that works without the need for corrections described in \citep{Patterson:2006,shriner2012improved}.  Our formulas show that our accuracy is not significantly impacted by reduction of the number of markers $N$, making thinning a useful technique for correcting for strong LD.

\textcolor{black}{Formula \eqref{eta} is informative about the size of the dataset required to detect substructure. Namely, it relates the theoretical magnitude of the eigenvalue to the dataset sample sizes of $M$ individuals and $N$ markers, for a fixed value of $\lambda _K$ (corresponding to the theoretical population separation). }
\textcolor{black}{For example,   when $M=1,000$ individuals and $N=500, 000$ markers,
we have $L=3664.9 \cdot \lambda _K$. Using this equation, we can explore how reducing the size of the dataset impacts the ability to detect substructure. Thinning the number of markers   to $N'=100, 000$, gives $L=3305.8 \cdot \lambda _K$.  A more stringent 10-fold thinning to $N''=50, 000$,  gives $L=3070.2 \cdot \lambda _K$, so if $\la_K$ is, say, larger than $0.0002$, thinning will not have much influence on the resolution.  To illustrate how resolution is affected by the number of individuals, we remark that the effects of the thinning of markers in the first example can countered by  increasing the number of individuals from $M = 1,000$ to $M''=1,121$, and the effects of the 10-fold thinning in the second example can be reversed by increasing the number of individuals to $M'' = 1,226 $.
 Thus, if $\la_K$ is far enough from zero, and $M,N $ are large enough so that $L$ is much larger than 0.5,  then a moderate thinning of the number of markers will have no effect  on accuracy. On the other hand, formula \eqref{eta} illustrates that  even slight increase in $M$ can compensate for a fewer number of markers, in this commonly encountered scenario where $N \gg M$.
}

A possible violation of the assumption of stochastic independence of individuals is non-random mating, which results in departures from Hardy-Weinberg equilibrium (HWE).
We find that departures from HWE do not significantly  reduce the power of PCA for detecting population substructure. In fact, to compensate for $F=1$ instead of $F=0$ it is enough to increase the number of individuals $M$ in the study by a factor of 2;  however, no such corrections are needed if the smallest eigenvalue $\la_K$ is separated from zero well enough so that $L>0.5$. %

Lastly, individuals that are closely related violate our assumption of random sampling of individuals from a subpopulation. These hidden, or ``cryptic'', relationships among individuals may affect the applicability of our method for population structure by changing the   distribution of eigenvalues. \textcolor{black}{Previous studies have shown hidden relatedness in the International HapMap Project (HapMap) data  \citep{Pemberton2010,Stevens2012}. PCA of the HapMap genotype data illustrates how the distribution of the small eigenvalues is disturbed by cryptic relationships}. Similar comments have been made by other authors; in particular \cite[page 2089]{Patterson:2006} warn about ``the inclusion of samples that are closely related".
 Under a simple substructure scenario, the \textcolor{black}{histogram distribution of the small} eigenvalues should have a unimodal elliptical shape similar to the Marchenko-Pastur distribution, easily distinguished from large eigenvalues corresponding to substructure. However, as we demonstrate in Figure \ref{Fig-MP-sim}, individuals may exhibit cryptic relatedness, or other unknown non-random relationships, which result in changes to the distribution of the bulk, making it difficult to infer the correct cutoff for substructure.  We find that pruning for LD does not seem to improve the fit of the bulk to Marchenko-Pastur distribution. Instead, exclusion of related individuals improves fit of the bulk; hence, we suggest that it is necessary to remove related individuals from the sample to improve the resolution of true substructure.

\subsection{Summary}Overall, our mathematical analysis confirms empirical evidence that PCA is a robust technique for learning about population substructure of a dataset. Contrary to current practice, based  on the mathematical theory presented in the following sections we recommend using PCA directly on the data matrix $\mC$ without centering or renormalization. \textcolor{black}{Since we do not center, we obtain $K$ large eigenvalues in the presence of $K$ subpopulations, instead of $K -1$ large eigenvalues when centering. This difference is proven in section 4.3. In avoiding renormalization of the data, we are able to produce mathematical theory showing that with} sufficient sample size, there should be strong separation between the large eigenvalues corresponding to population structure and the remaining bulk of the distribution.  We illustrate a proof of principle of our approach through simulations and application to human genotype data from world-wide populations.

\section{Results and Discussion}
In this section we illustrate the power of our mathematical findings for inference of population structure in genetic data. We begin with the simulations for a ``simple model" where all the hidden parameters can be computed. This allows us to analyze sensitivity of the technique to the precise value of $L$ in \eqref{eta}. In particular, since we are able to compute the hidden parameter $L$, we can then see how well our predictions match theory, and how well powered we are to detect the known substructure.  Then we consider an intermediate  stage -- we use simulated data from  \citep{gao2011identifying} where the true demography is known and each individual is a member of one of the subpopulations, but for which the hidden population parameters are not known. The datasets cover several different demographic models, with different subpopulation split times, trees, and migration rates. For more details on each of the models, see reference~\citep{gao2011identifying}. Finally, we apply the theory to human genotype data from world-wide populations, where we discuss additional challenges due to linkage disequilibrium and cryptic relatedness, and where the value of  $L$ is not available.

\subsection{Simulations for a simple model} \label{Sec:IP}
We show that the theoretical approximations to the largest eigenvalues work very well when all the assumptions of mathematical analysis are satisfied. We generate simulations based on a simple model in which we make several assumptions that are unrealistic, but allow us to compute important mathematical parameters to explore the performance of our method. We assume that the  site frequency spectra are known for each subpopulation. We also know how many individuals came from each subpopulation, and that the subpopulations are independent.
The latter corresponds to a scenario where all subpopulations  diverged and stopped interacting in the distant past.
 Though unrealistic, such a simplistic model has the advantage that all relevant quantities that enter mathematical analysis can be computed. In particular, this model allows us to study the effects of choosing a small enough number of individuals $M$ and analyze how the failure rate for the estimator depends on the value of  $L$, see Table \ref{T-eta-K}.

In our simulations we use unequal subpopulation samples sizes, drawn with proportions $c_1=1/6$, $c_2=1/3$, $c_3=1/2$. The theoretical population proportion $p_r(j)$ at  each SNP location for each subpopulation was selected from the
 same probability density function  $\psi(x)=0.5/\sqrt{x}$ (the so-called site frequency spectrum, see Section \ref{Sec:Th} and \cite{Kimura:1964}). %
 We selected $p_1(j), p_2(j),p_3(j)$ independently at each location $j$ %
 which corresponds to  product %
 joint probability density  $\psi(x,y,z)=\psi(x)\psi(y)\psi(z)$ for our $K=3$ simulated subpopulations.
 We then simulated independent individual genotypes for the   $j$-th marker   of a member of the $r$-th subpopulation by choosing independent binomial values  (with $2$ trials)  with probability of success $p_r(j)$.
    We can evaluate how well the mathematical description matches the simulated data  because we can explicitly compute the theoretical matrix of moments (\ref{Qrr}-\ref{Qrs})  and hidden matrix  $\mQ$  defined by \eqref{QQQ}: %
$$
[m_{r,s}]=\left[
\begin{array}{ccc}
 {1}/{5} & {1}/{9} & {1}/{9} \\
 {1}/{9} & {1}/{5} & {1}/{9} \\
 {1}/{9} & {1}/{9} & {1}/{5} \\
\end{array}\right],\;
\mQ=\left[\sqrt{c_rc_s}\,m_{r,s}\right]=
\left[
\begin{array}{ccc}
    0.0333&    0.0262&    0.0321 \\
    0.0262  &  0.0667&    0.0454 \\
    0.0321&    0.0454 &   0.1000 \\
\end{array}\right]
$$

The eigenvalues of the above matrix    $\mQ$ are  $[\la_1,\la_2,\la_3]= [ 0.1467, 0.0355, 0.0178]$.
The theoretical prediction for the observed largest  eigenvalues of the normalized  sample covariance matrix  \eqref{X} are then   given by \eqref{La-predicted}. %
    The actual observed eigenvalues will not match exactly these predictions; the purpose of the simulations is to illustrate how far the empirical values for finite $M,N$ differ from the values predicted by theory in the limit as $M,N$ tend to infinity.
For example,  formula \eqref{La-predicted}  with
$M=120, N=2500$
gives  the following values:
$(  47.4, 11.5, 5.7 )$.  In a simulation, we obtained  the following eigenvalues for the normalized matrix \eqref{X}:
$$({\bf\La_1,\La_2,\La_3},\La_4,\La_5,\dots ) =   ({\bf 48.2,   11.5,    5.8},    0.27,    0.26,\dots)$$
We see that the threshold of $0.5$ separates  clearly the $K=3$  largest  eigenvalues, set in boldface, from the bulk.

We remark that there are two sources for the discrepancy between theoretical and observed eigenvalues:
the approximation $\mB_N\approx \mQ$ that appears in Lemma \ref{L:asymp} and then the approximation due to randomness within each subpopulation that  is still present  in \eqref{Large_La} for finite numbers of SNPs, $N$.
It is therefore encouraging to see that the predicted values for the eigenvalues match well with the empirical eigenvalues in the simulations  for realistic values of $M=$ 120, $N= 2500$ as well was for much smaller values of $M$,  see Table \ref{T-eta-K}, %
or even for $M=12$. %
When $M=24$ and  $N=100$, we get $L=0.77$ %
and we are successful in determining correct value of $K=3$ the vast majority of the time: in 100,000 simulations,  $K$ is underestimated in only a minuscule  0.005 \% of the runs and never overestimated.    Reducing further the number of individuals  to $M=12$, leads to $L=0.471$ and in this case,  as expected, the rate at which our  estimate of $K$ fails increases. But the decrease in accuracy is not dramatic, and  an underestimate of $\hat K=2$ occurs only in about  6.5\% of runs. For reasonably large $M$ and $N$, our power is quite high, and
an overestimate of $K$ did not occur in any of the replicates under any scenario.

\begin{table}[H]\caption{False negative probability as a function of $L$ (based on 100,000 simulations) \label{T-eta-K}}
\begin{tabular}{cccll}
$M$ & $N$   &$L$ &$\Pr(\hat K<3)$ \\
 \textit{individuals} & \textit{SNPs} & & \textit{False negative rate}\\ \hline
120 & 2500  &5.747 & 0 \\
48& 100& 1.926  &0 \\
24& 100& 0.770   &    0.00005 \\
12&100& 0.471  &0.065 \\
12  & 50& 0.385   &0.53 \\
6&  100 &  0.276   &0.87
\end{tabular}
\end{table}

\subsection{Simulated genetic data under various demographic scenarios}
Next we applied our method to simulated substructured datasets generated by coalescent simulations under various demographic scenarios from   \citep{gao2011identifying}. This dataset has $N=100$ markers sampled from subpopulations with constant subpopulation sizes of 50 individuals, and with  varying $M=50\times K$, $K=1,\dots,5$.

In these simulations, $\mQ$ is not known.  But subpopulation labels are known, so we can estimate   $L$  by using the smallest eigenvalue of the empirical approximation to $\mQ$ based on formula  \eqref{Qhat} with allelic probabilities estimated from \eqref{phat}.
The observed accuracy under each scenario, shown in Tables~\ref{Hong2}--\ref{Hong4}, depends on the estimate $\hat L$ of $L$  consistently with our results from simulations listed in Table~\ref{T-eta-K}.

Model-based approaches such as those evaluated in \citep{gao2011identifying} are likely to outperform PCA detection for such small sample sizes, since the true substructure corresponds to that found in the underlying \textit{STRUCTURE}-like model \citep{Falush2003}.
However, using our method, we have no overestimates of the number of subpopulations $K$ %
in any of these sets of simulations.
From these simulations we find that the error rates are not affected by using approximate $\mQ$ in the calculations instead of exact $\mQ$ when we approximate $L$ from the subpopulation data.
(The values $ \hat L$ of the approximated $L$ varied considerably in the simulated 50 runs for a model, but $\hat L >0.5$ was  associated with the correct value of $\hat K$ in each case.)

  \begin{table}[H]
 \begin{tabular}{r||cccc}
 True $K$  & 2 & 3& 4&5 \\
 $\Pr(\hat K<K)$
 & 0.0
 & 0.14
 & 0.80
 &  0.98
  \end{tabular}
 \caption{False negative (error) rates of estimates of $K$ for 50 simulated data sets under model \textit{Split} \label{Hong2} }%
 \end{table}

 \begin{table}[H]
 \begin{tabular}{r||cccc}
True $K$  & 2 & 3& 4&5
  \\
 $\Pr(\hat K<K)$
 & 0.0
 & 0.16
 & 0.72
 &  1.0
  \end{tabular}
 \caption{False negative (error) rates of estimates of $K$ for 50 simulated data sets under model \textit{Inbred} \label{Hong3} } %

 \end{table}

    \begin{table}[H]
 \begin{tabular}{r||cccc}
True $K$  & 2 & 3& 4&5
   \\
 $\Pr(\hat K<K)$
 &  0.0
 &  0.02
 &  0.10
 &  0.46
  \end{tabular}
 \caption{False negative (error) rates of estimates of $K$ for 50 simulated data sets under model \textit{Mig} \label{Hong4}} %

  \end{table}

\subsection{Application to human population genotype data}

 We next examined the distribution of eigenvalues for a dataset of human genotype data.
The International HapMap Project was designed to create a catalog of human genetic variation to find genes that affect health, disease, and individual responses to medications and environmental factors. We use a genome-wide SNP dataset made publicly available through this project as HapMap 3
(HapMap 3, release 3, human genome build 36)
which contains genotypes of individuals from 11 human populations, comprised of genotype data collected using two
   platforms: the Illumina Human1M and the Affymetrix SNP 6.0 arrays. These populations and datasets have been extensively studied previously (see \small{\texttt{http://hapmap.ncbi.nlm.nih.gov/publications.html.en}} \normalsize for a list of publications).

Unlike simulated data, the true substructure of the complete set of populations is unknown.  We therefore report the performance of our theoretical analysis on the subset of well defined subpopulations, which are believed to have clear substructure: the Yoruba, of Ibadan, Nigeria (YRI), European Americans from Utah (CEU), and Han Chinese from Beijing, China  (CHB).

After extracting the CEU, CHB, and YRI individuals,  we processed the data through PLINK \citep{purcell2007plink}  with   filters
{\tt --filter-founders --geno 0} to remove SNPs with any missing data and exclude offspring of trios, and further exclude non-autosomal markers. The final dataset of  $M=297$ individuals and  $N=736750$ markers was used for the analysis of the eigenvalues.

As expected from mathematical theory and from the choice of very distinct subpopulations, the eigenvalues of matrix $\mX$ split  into \textcolor{black}{two sets: the non-significant, or small, eigenvalues} in Figure~\ref{Fig3pop} that lie below the cutoff of 0.5, and three large eigenvalues $\La_1=102.0$,  $\La_2=14.55$, and $\La_3 = 7.37$ that exceed the cutoff of $0.5$ and give an estimate of $\hat{K}=3$, which matches our prediction for these three populations. The histogram seems to show some possible eigenvalues separated from the bulk, but these may correspond to various minor deviations from the model that are present in real data - we discuss this issue below.

  \begin{figure}[H]
  \begin{center}

    \includegraphics[width=10cm]{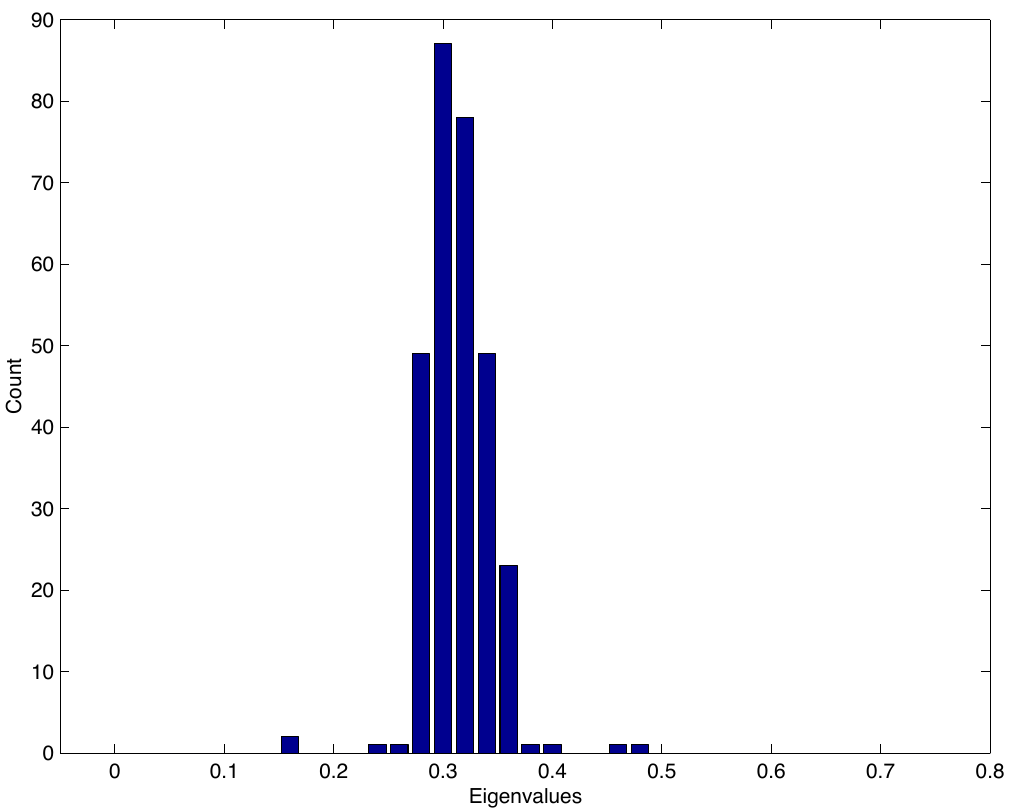}
   \caption{\textcolor{black}{Histogram of the eigenvalues from PCA of Hapmap CEU, CHB, and YRI unrelated individuals (parents of trios), excluding the large eigenvalues ($ \Lambda \gg 1$), which are omitted to better illustrate the shape of the non-significant eigenvalues. Here, the largest three eigenvalues that correspond to subpopulation structure are $\La_1= 102.0$, $\La_2= 14.55 $, $\La_3= 7.37$.}
\label{Fig3pop}   }
   \end{center}

  \end{figure}

\subsection{Practical comments on using theory} %
Theorem \ref{T1} and the cutoff of $0.5$ should be used in real data only after visual control for the shape of the \textcolor{black}{distribution of the eigenvalues} and for the separation \textcolor{black}{ of the non-significant eigenvalues, or the ``bulk'', from the largest eigenvalues}. Simulations indicate that when the theory is applicable the histogram of the bulk is located to the left of 0.5 (or 1 when $F=1$) and  its shape resembles the Marchenko-Pastur law %
of the same ratio $N/M$.  For a typical large value of $N/M>50$, this shape  looks similar to a semi-ellipse.

The shape of the \textcolor{black}{histogram of the distribution of eigenvalues} is  affected by relationships between the individuals. This is best illustrated when the offspring of trios are included in the  analysis of  the three populations HapMap dataset. Then the shape of the \textcolor{black}{distribution} does not follow Marchenko-Pastur law, and instead resembles a shape reproduced by repeating a large number of individuals, see Figure~\ref{Fig-MP-sim}.
  \begin{figure}[H]
  \begin{center}
\begin{tabular}{cc}
  Simulation of dependent  individuals&  Hapmap trios with offspring\\
   \includegraphics[width=6cm]{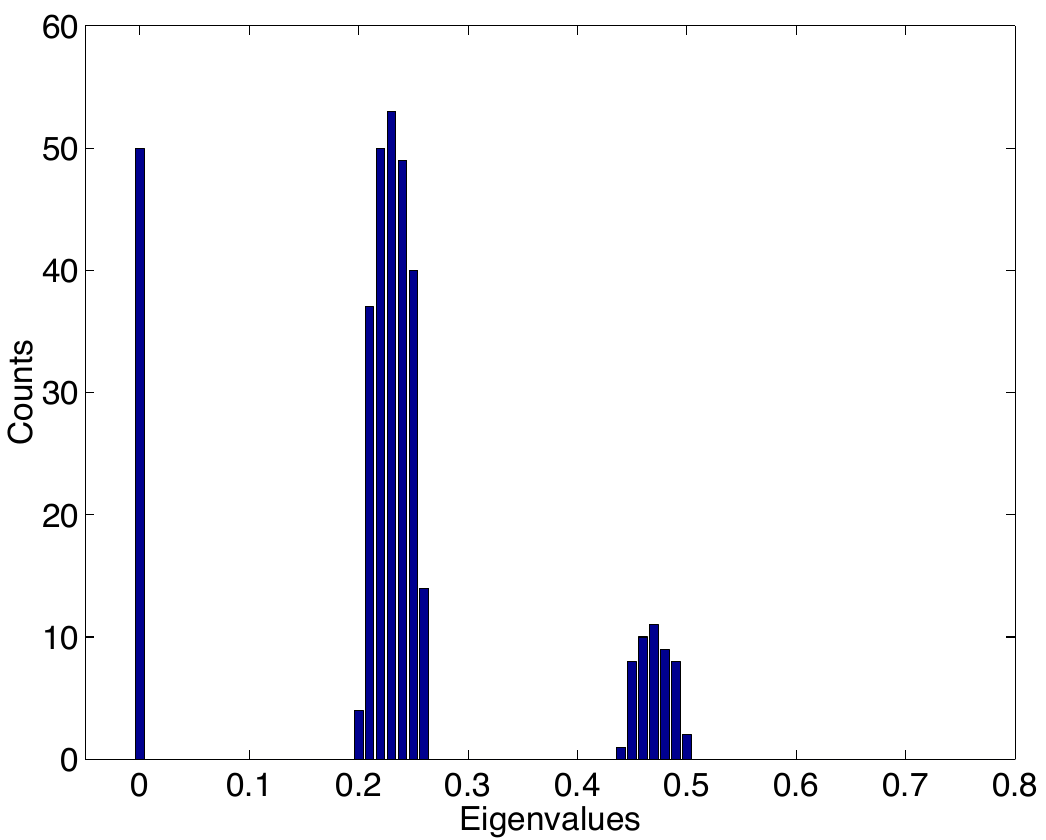}
    &   \includegraphics[width=6cm]{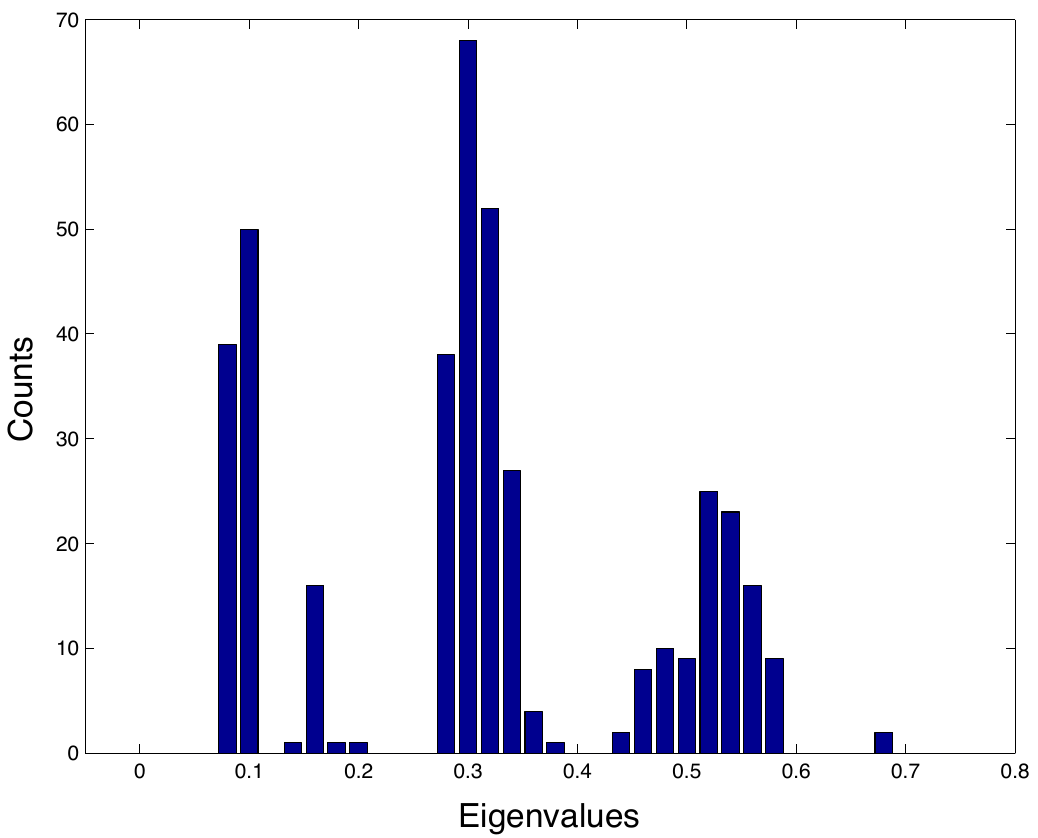}   \\
$M=350$, $N=70,000$.
& $M=405$, $N=660,847$. %
 \end{tabular}
 \vspace{6pt}
\caption{\textcolor{black}{Including related individuals perturbs the expected distribution of eigenvalues resulting from PCA.  \textit{Left:} A counts histogram of the eigenvalues from PCA using data generated via binomial simulation, where 29\% of the individuals have been repeated. Large eigenvalues (corresponding to population structure) are not shown so as to better illustrate the effect on the distribution of the non-significant eigenvalues. \textit{Right:}  A histogram of the eigenvalues for PCA of three populations of HapMap (CEU, YRI, and CHB) including trios -- 297 parents and their related 108 offspring. Large eigenvalues are not shown.   Both simulated data and empirical genotype data show that inclusion of related individuals results in a multi-modal distribution of the eigenvalues, arising from the non-random correlations of individuals. } }  \label{Fig-MP-sim}
\end{center}
\end{figure}

The shape of the \textcolor{black}{histogram of eigenvalues} for the full  HapMap data set  seems to exhibit additional  deviation from the expected shape, extends well beyond  0.5, and the distribution of the \textcolor{black}{small eigenvalues} might  not  be unimodal.   These deviations cannot be attributed to linkage disequilibrium as they do not disappear after LD pruning. Instead, we expect these deviations from the expected shape correspond to complex substructure and relationships among individuals\textcolor{black}{, as has been suggested in previous studies of cryptic relationships among HapMap samples \citep{Pemberton2010,Stevens2012} }.

  \begin{figure}[H]
\begin{center}
   \includegraphics[width=10cm]{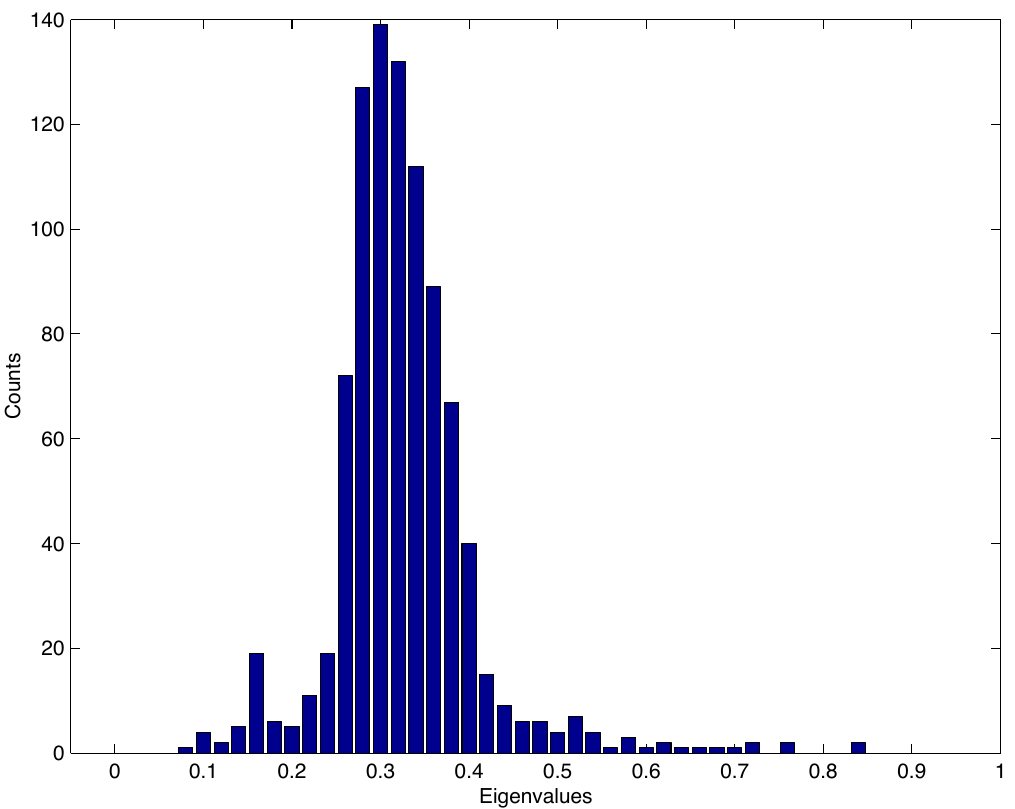}

\caption{
\textcolor{black}{Histogram of the eigenvalues from PCA of all 11 populations in HapMap unrelated individuals.
 Nonautosomal markers with $M=924$, $N=422253$.
The six largest eigenvalues  $\La_1=335.9$, $\La_2=37.4$, $\La_3=16.7$,  $\La_4=2.5$,  $\La_5=2.1$, $\La_6=1.7$ are  not  shown. }
   \label{Fig-hapmap11pop} }
\end{center}
\end{figure}

\section{Theory} \label{Sec:Th}
Our goal in this section is to point out  aspects of population structure that could be responsible for  the observed phenomenon that the set of  eigenvalues of $\mC\mC'$  splits into two groups: a small set of $K$ large eigenvalues, and  a large set of $M-K$ of  small eigenvalues. Since $\mC\mC'$ is a random matrix, we want this split to occur with overwhelming probability. This task requires a more detailed specification of the model. While the statements become more cumbersome, the gain is a clear  indication of how different  aspects of the model influence our ability to discover the subpopulation structure, and under what circumstances it may remain hidden.

As explained in Section \ref{Sec:patterson}, we assume that our genetic markers are biallelic and that {we have} $N$ polymorphic markers.  We {assume that we have $M$  diploid individuals} from $K$ subpopulations, and that
 the genotype probabilities for marker $j$ of  individual $i$  from subpopulation $r$
are  described by formulas
\eqref{Hardy-Weinberg}-\eqref{Hardy-Weinberg3}.
  The allelic probabilities $p_r(1),\dots,p_r(N)$ for the $r$-th subpopulation are unknown but represent the true underlying frequency of the alleles in the $r$-th subpopulation, and they are fixed when sampling the individuals from the subpopulation.

In diffusion models for a single population, allelic probabilities $p_r(1),\dots,p_r(N)$ are considered random and are then adequately described  by their density function $\psi_r(x)$, $x\in[0,1]$, see e.g.  \citep{Kimura:1964}. We shall call $\psi_r(x)$  the (univariate) site frequency spectrum for the $r$-th population.
A site frequency spectrum is informative of the demographic history of a set of samples, and joint site frequency spectra have been used in multi-population demographic analysis \citep{Gutenkunst2009,xie2011site,bustamante2001directional}. The site frequency spectrum is a theoretical distribution capturing all the genetic variation present in a set of individuals. However,
 obtaining a site frequency spectrum requires high quality sequence data, or computational correction of the method of SNP discovery resulting in SNP ascertainment bias \citep{Keinan2007, Clark2005}.   For our analysis we do not need the site frequency spectra; instead, we only require some set of informative markers to model each pair of subpopulations $r,s$ in terms of its pairwise probability density function $\varphi_{r,s}(x,y)$. Our requirement is such that the limit \eqref{kkk} exists, so our model allows for any distribution of markers, and we can perform analysis without correction on SNP genotype data that can be of unknown or complex ascertainment.

Consequently, we assume that  allelic probabilities $p_r(1),\dots,p_r(N)$ are random and are  adequately described  by their density function $\varphi_r(x)$,   and that for the $j$-th locus each pair of allelic probabilities $(p_r(j) , p_s(j))$  follow the same bivariate distribution with density $\varphi_{r,s}(x,y)$. This is a  natural setting where the limit  \eqref{kkk} exists, and if we know the ascertainment-biased  pairwise site frequency spectra, the limit is given by the pairwise subpopulation moments:

 \begin{equation}\label{Qrr}
 m_{r,r}=\int_0^1 x^2\varphi_r(x)dx,
 \end{equation}
and  for $r\ne s$,
\begin{equation}\label{Qrs}
m_{r,s}=\int_0^1\int_0^1xy\,\varphi_{r,s}(x,y)dxdy.
\end{equation}

Several authors   \citep{balding1995method,Patterson:2006,Pritchard2000} avoid complications of ascertainment biased site frequency spectrum by  using a different model. They  they assume that for a single marker the allelic probabilities $p_1,\dots,p_K$  evolved from a single ancestral population with diffuse marker distribution $\psi ( {p}) dp$ and conditionally on $p$ have mean $p$ and covariance $cov(p_r,p_s|p)=p(1-p)B_{rs}$, where $B=[B_{r s}]$ is a positive definite $K\times K$ matrix.
 When $N$  independent  markers are drawn from this model, by the law of large numbers
$$
\lim_{N\to\infty}\frac{1}{N} \sum_{j=1}^N p_r(j)\to \int_0^1 p \psi({p}) dp
$$
and
\begin{multline*}
\lim_{N\to\infty}\frac{1}{N}\sum_{j=1}^N p_r(j)p_s(j)\to \int_0^1 E(p_r p_s|p) \psi({p})dp\\ =
 \int_0^1 (cov(p_r p_s|p)+p^2)  \psi({p})dp  = B_{rs}\int_0^1 p \psi({p})dp + (1-B_{rs}) \int_0^1 p^2 \psi({p})dp
\end{multline*}
Thus our assumption \eqref{kkk} holds with  $m_{rs} = B_{rs}\mu_1 + (1-B_{rs})\mu_2$ replacing \eqref{Qrs}, where $\mu_1=\int_0^1 p \psi({p})dp$ and $\mu_2= \int_0^1 p^2 \psi({p})dp$.  This shows that our results on separation of the largest eigenvalues are applicable  to this model as well.

Next we turn to assumptions on data matrix $\mC$. %
Since  the singular values of $\mC$ do not depend on the order of the rows, for mathematical analysis we assume that the individuals were arranged by subpopulation, so that $\mC$ has block structure  \eqref{Block}.
  \begin{equation}\label{Block}
  \mC=\left[\begin{matrix}
  \mC_1\\\mC_2\\\vdots\\\mC_K
\end{matrix} \right]
\end{equation}
where $\mC_r$ is the $M_r(N)\times N$ sub-matrix representing the data for the individuals from the $r$-th subpopulation.
For  $r=1,\dots,K$, $i=1,\dots,M_r$, $j=1,\dots,N$, we assume that the distribution of entry $[\mC_r]_{i,j}$ is given by  \eqref{Hardy-Weinberg}-\eqref{Hardy-Weinberg3}.  We assume that all the entries of matrix $\mC=\mC_N$ are independent, conditionally on $\{p_{r}(j)\}$.

For the almost sure results we also need to assume that $\mC$ comes from an infinite matrix.  More specifically, we assume that conditionally on $\{p_{r}(j)\}$, each of the $K$ blocks $\mC_r$ arises as an upper-left $M_r\times N $ corner of an infinite matrix with independent entries that have distribution \eqref{Hardy-Weinberg}--\eqref{Hardy-Weinberg3}.  %
We also need to make a technical assumption that the number of individuals from the $r$-th subpopulation  increases with $N$, that is $M_r(N+1)\geq M_r(N)$.

Since the  eigenvalues of  $\mC\mC'$  are large, it is more convenient to  consider the normalized $M\times M$ sample covariance matrices \eqref{X}.
Of course, we assume $N>M>K$.

\begin{theorem}\label{T1}
Let $\la_1\geq\la_2\geq\la_K$ be the  (deterministic) eigenvalues of matrix $\mQ$  defined by \eqref{QQQ} and let $\La_1(N)\geq\La_2(N)\geq\dots\geq\La_{M(N)}(N)\geq 0$ be the (random) eigenvalues of the $M(N)\times M(N)$ sample covariance matrix $\mX_N$ from \eqref{X}.
Then, as $N\to\infty$, with probability one
\begin{equation}\label{MP}
\La_{K+1}(N)\leq  (1+F)/2
\end{equation}
and
 \begin{equation}\label{Large_La}
\left(\frac{1}{\sqrt{M(N)}}+\frac{1}{\sqrt{N}}\right)^2\left[\La_1(N),\La_2(N),\dots,\La_{K}(N)\right]\to
4\left[\la_1,\la_2,\dots,\la_{K}\right]\end{equation}

\end{theorem}
When $\mQ$ has full rank $K$ formula \eqref{Large_La} indicates that for large $M,N$  the first $K$ largest empirical eigenvalues of $\mX_N$ are large and can be   estimated from the eigenvalues of $\mQ$. Formula
\eqref{MP} shows that the remaining eigenvalues are relatively small, and are of the order $1/M$ smaller than the largest $K$ eigenvalues.
\begin{remark}\label{linkage}
Under Hardy-Weinberg equilibrium $F=0$, so \eqref{MP} takes form
\begin{equation}\label{MP0}
\La_{K+1}(N)\leq  1/2.
\end{equation}
However, usually the value of $F$ is not known.
In such cases,  since $0\leq F\leq 1$,  while $\La_{K}(N)\to \infty$ as $N\to\infty$ is much larger than 1,  formula \eqref{MP} may be replaced by
\begin{equation}\label{MP1}
\La_{K+1}(N)\leq  1.
\end{equation}
\end{remark}
\begin{remark}
We have much less information for the case when  $\mQ$ has rank $K'<K$ with positive  eigenvalues $\la_1\geq\la_2\geq\la_{K'}>0$ but $\la_{K'+1}=\dots=\la_K=0$.  In this case the bulk is still concentrated below 1/2, as   \eqref{MP} shows that $\La_{K+1}\leq 1/2$. From
\eqref{Large_La} we deduce that  the eigenvalues $\La_1,\dots,\La_{K'}$ diverge to infinity. But
we do not have any information about $\La_{K'+1},\dots,\La_{K}$ which in this case are  only  known to be of order smaller than
$\frac{MN}{(\sqrt{M}+\sqrt{N})^2}$; we do not have any mathematical results  about  their relation  to the ``cutoff" $(1+F)/2$.

\end{remark}

\subsection{Proof of Theorem \ref{T1}} %
From   \eqref{Hardy-Weinberg}-\eqref{Hardy-Weinberg3} it follows that the expected value of the entry  $C_{i,j}$ for an individual from the $r$-th subpopulation is $2 p_r(j)$ and the variance   is $2(1+F_{r,j})p_{r}(j)(1-p_{r}(j))\leq  (1+F)/2$.

Let $P_r\in\RR^N$ denote the column vector  $[p_r(1),\dots,p_r(N)]'$.
Let $E_r\in\RR^M$ be the column vector  of ones at the rows corresponding to the $r$-th block of $\mC$, i.e., with $[E_r]_i=1$ if $M_1+\dots+M_{r-1}< i\leq  M_1+\dots+M_{r}$ and $0$ otherwise. Then $\E(\mC)=2\sum_{r=1}^K E_r P_r'$, and
we write
\begin{equation}\label{Eq:perturb}
\mC_N=\mV+2 \sum_{r=1}^K E_r P_r'\;,
\end{equation}
where  $\mV$ is an $M\times N$ matrix of centered independent uniformly bounded random variables. Note that  $\E(\mC)$ factors as in  \cite[Eqn. (1)]{Engelhardt2010}, but we keep an additional term in analyzing    \eqref{Eq:perturb}.

Let $\la_1\geq\dots\geq \la_K\geq 0$ be all eigenvalues of $\mQ$. (Recall that  $N> M=M(N)>K$.)
\begin{lemma}\label{L:asymp}
Let $\sigma_1(N)\geq\sigma_2(N)\geq\dots\geq\sigma_K(N)\geq 0$ be the singular values of $\sum_{r=1}^K E_r P_r'$.
Then for $1\leq r\leq K$,
$\lim_{N\to\infty}\frac{\sigma_r^2(N)}{NM(N)}=\la_r$.
\end{lemma}
\begin{proof}

Consider the sequence of $K\times K$ matrices $\mB_N$ with entries
\begin{equation}\label{Qhat}
[\mB_N]_{r,s}=\frac{\sqrt{M_rM_s}}{MN}\sum_{j=1}^N p_r(j)p_s(j).
\end{equation}
(Recall that $M_r=M_r(N)$ is a function of $N$.)
Since $\mB_N\to \mQ$ entrywise,   its eigenvalues $\la_1(\mB_N),\dots,\la_K(\mB_N) $ converge to $\la_1,\dots,\la_K$.
 However, $\la_r(\mB_N)=\frac{\sigma_r^2(N)}{NM}$ for all $r\in\{1,\dots,K\}$.
 To see this,
denote $U_k=\frac{1}{\sqrt{M_k}}E_k$. Then
$$\left(\sum_{r=1}^K E_r P_r'\right)\left(\sum_{s=1}^K E_s P_s'\right)'
=NM\sum_{r,s=1}^K [\mB_N]_{r,s}U_rU_s'$$
Let now $\vec{x}=[x_1,\dots,x_K]'$ be an eigenvector corresponding to eigenvalue $\la$ of $\mB_N$. Then $\vec{y}=\sum_{r=1}^K x_r U_r\in\RR^M$ is an eigenvector  of
$\sum_{r,s=1}^K [\mB_N]_{r,s}U_rU_s'$ with the same $\la$. So $NM\la$ is the square of a singular value of   $\sum_{r=1}^K E_r P_r'$.
Note that orthogonal vectors $\vec{x}$ correspond to orthogonal $\vec y$, so this procedure exhausts the first $K$ eigenvalues, even if they are repeated or $0$; the remaining $M-K$ eigenvalues are zero .
\end{proof}

\begin{lemma}\label{L:norm} With probability one, as $M/N\to c>0$,
$$\limsup_{N\to\infty}\frac{1}{\sqrt{M}+\sqrt{N}}\|\mV_N\|\leq \sqrt{\frac{1+F}{2}}.$$
\end{lemma}
\begin{proof}
We apply a non-i.i.d. version of \cite[Theorem 3.1]{Yin:1988},
which was  extended in \cite[Theorem 3]{Couillet:2010} to allow for the distributions of the entries to vary. Specifically, we apply \cite[Eqtn. (94)]{Couillet:2010}.
To do so, we need to verify that the entries of matrix $\mV/\sqrt{(1+F)/2}$  satisfy assumptions (1)-(6) in part A of the proof on  \cite[page 2437]{Couillet:2010}.

Conditions (1) and (3) hold by assumption.
The entries of   matrix  $\mV/\sqrt{(1+F)/2}$   are uniformly bounded in absolute value by  the constant $C=2\sqrt{2}/\sqrt{1+F}$. Thus with  $\eta_n=C/\sqrt{n}$, condition (2) is verified.
 Condition (6) is then automatically satisfied with $c=C^3$.
For $i$ from the $r$-th block $M_1+\dots+M_{r-1}< i\leq  M_1+\dots+M_{r}$ equations \eqref{Hardy-Weinberg}-\eqref{Hardy-Weinberg3}  give  \begin{equation}
\label{NoLink}
\E\left([\mV]_{i,j}^2\right) =2(1+F_{r,j})p_{r,j}(1-p_{r,j})\leq ( 1+F)/2,
\end{equation}
So   condition (4) holds.

To verify assumption (5) we use the pointwise  estimate $\left|[\mV]_{i,j}\right|^\ell< 2^{\ell-1}\left|[\mV]_{i,j}\right|  $ followed by the Cauchy-Schwartz inequality:
$$\E\left(\left|[\mV]_{i,j}\right|^\ell\right)< 2^{\ell-1} E\left(|[\mV]_{i,j}|\right)\leq 2^{\ell-1} (E|[\mV]_{i,j}|^2)^{1/2}\leq 2^{\ell-3/2} \sqrt{1+F}.$$
So
$$\frac{\E\left(\left|[\mV]_{i,j}\right|^\ell\right)2^{\ell/2}}{(1+F)^{\ell/2}}
\leq
2^{\ell-1}2^{(\ell-1)/2}/(1+F)^{(\ell-1)/2}=
\left(2\sqrt{2}/\sqrt{1+F}\right)^{\ell -1}$$
Recalling that    $2\sqrt{2}/\sqrt{1+F}=\eta_n\sqrt{n}$, this becomes condition (5) for the matrix  $\mV/\sqrt{(1+F)/2}$.
%(We note that strict inequality is not needed for the proof in \cite{Yin:1988}, and can be achieved by replacing $\eta_n$ by $(1+1/n)\eta_n$.)

By Borel-Cantelli lemma, \cite[Eqtn. (94)]{Couillet:2010}  implies that $\limsup_{N\to\infty}\frac{\|\mV\|}{\sqrt{M}+\sqrt{N}}\leq\sqrt{(1+F)/2}$.

\end{proof}

\begin{proof}[Proof of Theorem \ref{T1}]
This part of the proof is similar to \citep{Silverstein:1994}.
 In \eqref{Eq:perturb}, we consider $\mC_N$  as a small
perturbation of  the finite rank matrix $2\sum_{r=1}^K E_r P_r'$.

Denote by $\tau_1(N)\ge \dots \ge\tau_K(N)$ the  largest (deterministic) singular values of $$\frac{2}{\sqrt{N}+\sqrt{M}}\sum_{r=1}^K E_r P_r',$$ and set $\tau_j(N)=0$ for $j>K$.
Then it is known, see e.g. \cite[Theorem 3.3.16(c)]{horn1994topics}, that
the singular values $\sqrt{\La_j}$ of $\frac{1}{\sqrt{N}+\sqrt{M}}\mC$, written in decreasing order, differ by  at most
$\frac{1}{\sqrt{N}+\sqrt{M}}\|\mV_N\|$ from the corresponding singular values  $\tau_j(N)$, written in decreasing order.

Since $\tau_j(N)= 0$ for $j> K$,  from Lemma \ref{L:norm} we get  \eqref{MP}.

Lemma \ref{L:asymp} shows  that %

$$\tau_j(N) \left(\frac{1}{\sqrt{M}}+\frac{1}{\sqrt{N}}\right)=\frac{\tau_j(N) (\sqrt{M}+\sqrt{N})}{\sqrt{MN}}=\frac{2\sigma_j(N)}{\sqrt{MN}}\to 2 \sqrt{\la_j}$$

 for $1\leq j\leq K$,
 so
 $$\sqrt{\La_j} \left(\frac{1}{\sqrt{M}}+\frac{1}{\sqrt{N}}\right)$$ has the same limit, and by taking squares of both sides we get \eqref{Large_La}.
\end{proof}

\subsection{An estimate of the sample size needed for significance}  %
Consider an example of two samples, each of size $M/2$, from two subpopulations with the same values $m_{1,1}=m_{2,2}$ in \eqref{kkk}. Then the smaller eigenvalue can be computed explicitly,  $\la_2=\lim_{N\to\infty}\frac{1}{2N}\sum_{j=1}^N(p_1(j)-p_2(j))^2$. So  \eqref{eta} gives  \begin{equation}\label{Transition}
\lim_{N\to\infty}\frac{1}{N}\sum_{j=1}^N(p_1(j)-p_2(j))^2>\frac{(1+\sqrt{d})^2}{4M}, \;  d=M/N
\end{equation}
as the bound for $M$ to detect the population structure.
This shows that any   difference can be detected  by increasing the number of individuals $M$, but it is also of interest to note that if $M$ is too small, then the  increase in  $N$ has only a  limited benefit of reducing the value of $(1+\sqrt{d})^2$ from $4$ to $1$.  Thus, a too-small value of $M$ cannot  be compensated for by increasing the number of markers $N$.  %

Simulations confirm  that   for a fixed $N$, the probability of detecting $K$    rises sharply from 0 to $1$ as $M$ increases.  For a given size of the data matrix, as measured by constant value of the product $M N$, simulations  indicate that when $M\leq N$ the larger values of $M$ increase the power in model  \eqref{Eq:perturb}, so this model behaves differently than the Wishart model discussed in  \cite[page 2083]{Patterson:2006}.
\subsection{Centered data have $K-1$ large eigenvalues}\label{Sect:Centering}

In this section we point out that our basic conclusions with appropriate modifications can be used to justify that for a  matrix $\bar{\mC}$ with centered columns  $K$ subpopulations  correspond to $K-1$  large singular values instead of $K$ large eigenvalues as    in Theorem \ref{T1}.
That is, for centered data  \eqref{mC-centered} coming from $K$  well separated  subpopulations, the centered  matrix $\bar{\mC}/(\sqrt{M}+\sqrt{N})$ has $K-1$ large singular values that grow without bound while  the remaining singular values remain bounded and are smaller than $\sqrt{2(1+F)}$.

This can be seen by adapting the arguments that we used in the proof of Theorem \ref{T1}. Denote by $\mathbf{1} $  the $M$ dimensional column vector consisting of all one's. The average of the columns  of
 $\mC_N = \mV_N + 2 \sum_{r=1}^K E_rP_r'$ is
 $$
 \frac{1}{M}\mathbf{1}'\mV_N+ \frac{2}{M} \mathbf{1}' \sum_{r=1}^K E_rP_r'
 $$
 So the centered matrix is
 \begin{equation}\label{mC-centered}
 \bar{\mC}_N= \left( \mV_N-\frac{1}{M}\mathbf{1}\mathbf{1}'\mV_N\right)+ 2 \sum_{r=1}^K \left(E_r-\frac{M_r}{M}\mathbf{1}\right)P_r',
\end{equation}

As in the proof of Theorem \ref{T1}  we interpret  $\bar{\mC}_N$  as a random perturbation of the finite rank matrix $2 \sum_{r=1}^K \left(E_r-\frac{M_r}{M}\mathbf{1}\right)P_r'$. We use the triangle inequality to bound the norm of the perturbation:
$$
\left\|\left( \mV_N-\frac{1}{M}\mathbf{1}\mathbf{1}'\mV_N\right)\right\|\leq 2\|\mV_N\|
$$
so in the limit the singular values of  $\frac{1}{\sqrt{M}+\sqrt{N}}\bar{\mC}_N$ differ by at most $ \sqrt 2\sqrt{1+F}
$ from the singular values of  matrix
$$A=\frac{2}{\sqrt{M}+\sqrt{N}} \sum_{r=1}^K \left(E_r-\frac{M_r}{M}\mathbf{1}\right)P_r'$$
We note that since $\mathbf{1}=\sum_{r=1}^KE_k$ we can rewrite
$$A = \frac{2}{\sqrt{M}+\sqrt{N}} \sum_{r=1}^K  E_r  \bar P_r' $$
where
$$\bar P_r=P_r- \sum_{s=1}^K \frac{M_s}{M}P_s.$$
Indeed,
\begin{multline*}
\sum_{r=1}^K \left(E_r-\frac{M_r}{M}\mathbf{1}\right)P_r'=\sum_{r=1}^K  E_rP_r' -\sum_{r=1}^K \frac{M_r}{M}\mathbf{1} P_r' \\
= \sum_{r=1}^K  E_rP_r' -\frac{M_r}{M}\sum_{r=1}^K\left(\sum_{s=1}^KE_s\right)P_r'
 =\sum_{r=1}^K  E_rP_r' -\sum_{s=1}^K\sum_{r=1}^K\frac{M_r}{M}E_sP_r'\\
  =\sum_{r=1}^K  E_rP_r' -\sum_{r=1}^K\sum_{s=1}^K\frac{M_s}{M}E_rP_s' =\sum_{r=1}^K  E_r\left(P_r' - \sum_{s=1}^K\frac{M_s}{M} P_s'\right)
\end{multline*}
 (We changed the order of summation in the second line, and swaped $r,s$ in the third line.)

We note that vectors $\bar P_1,\dots,\bar P_{K-1}$ are linearly independent as together with $\sum_{s=1}^K \frac{M_s}{M}P_s$ they span the same subspace  of $\RR^N$ as vectors $P_1,\dots,P_K$. The last vector, $\bar P_K=-\sum_{s=1}^{K-1} \frac{M_s}{M_K}P_s$, is a linear combination of the other vectors.  This shows that  the deterministic matrix $A$ has rank $K-1$ and that it has the  form appropriate to apply
   Lemma \ref{L:asymp}. The argument used in the proof of Theorem \ref{T1} shows that when  matrix $\bar \mQ$ corresponding to the moments of $\{\bar p_r(j)\}$ has $K-1$ positive eigenvalues, the singular values of $\bar{\mC}_N$ separate into $K-1$ large eigenvalues and remaining small eigenvalues as claimed.

\subsection{Test for presence of population substructure}
\cite{Patterson:2006}  pioneered the use of the Tracy-Widom distribution to test for presence of population substructure using the singular values of individual data.
Here we indicate how such a test can be justified for the data matrix $\mC$ which is modeled as a random perturbation of  a finite rank matrix \eqref{Eq:perturb} rather than as a spiked covariance model.

Under the null hypothesis  $H_0: K=1$ matrix $\mC$  has independent entries and the entries in the $j$-th column  are   independent and identically distributed  with  the expected value $2p(j)$. Under the Hardy-Weinberg equilibrium, the variances of the entries in the $j$-th column are all equal to $4p(j)(1-p(j))$. So under the null hypothesis,  the centered and normalized  matrix
 \begin{equation}
   \label{tilde-C}
  \widetilde C_{i,j}=  \frac{C_{i,j}-2 p(j)}{\sqrt{2 p(j)(1- p(j))}}
 \end{equation}
 has independent entries with mean zero and variance 1.

Due to the normalization by $\sqrt{2 p(j)(1- p(j))}$, to apply mathematical theory we need to  assume additionally  that  $p(j)$ are bounded away from $0$ and from $1$. %
Mathematical theory requires also  that $M/N\to d\ne1,0,\infty$.  With these assumptions in place, we can now apply results from   \cite{pillai2012edge}. Denote by $f_{\alpha}$ the $1-\alpha$ percentile of the Tracy-Widom distribution and let  $\tilde \Lambda_1$ be the largest eigenvalue of  $\widetilde{\mC}\widetilde{\mC}'$. According to  \cite[Corollary 1.2]{pillai2012edge},
\begin{equation}
  \label{Tracy-Widom}
\Pr\left(\tilde \Lambda_1>(\sqrt{M}+\sqrt{N})^2+f_\alpha (\sqrt{M}+\sqrt{N})\left(\frac{1}{\sqrt{M}}+\frac{1}{\sqrt{N}}\right)^{1/3}\right)\to \alpha
\end{equation}
as $N\to\infty$. Thus for large $N$ we  reject the null hypothesis on the level of significance $\alpha$ when $\widetilde\La_1/(\sqrt{M}+\sqrt{N})^2$ exceeds the threshold
\begin{equation}\label{TW-cutoff}
  1+\frac{f_\alpha}{ (\sqrt{M}+\sqrt{N})}\left(\frac{1}{\sqrt{M}}+\frac{1}{\sqrt{N}}\right)^{1/3} .
\end{equation}

This is of course a  different, and in some ways more precise statement than \eqref{MP}, which for the normalized case  would have said that as $N\to\infty$ with probability one  $\tilde\La_1/(\sqrt{M}+\sqrt{N})^2\to 1$ by \cite[Theorem 3]{Couillet:2010}.
On the other hand,    in \eqref{Large_La} we give additional information about the case when $K\geq 2$, showing   that  $\La_1/(\sqrt{M}+\sqrt{N})^2\to \infty$ will eventually exceed  any fixed cutoff; in Section \ref{Sect:Centering} we point out that a similar conclusion is available for the  squares $\bar{\La}_1,\dots,\bar{\La}_{K-1}$  of the largest  the  largest $K-1$  singular values of  the  centered matrix \eqref{mC-centered}.

Normalization \eqref{tilde-C} uses unknown allelic probabilities $p(j)$.  It is natural to conjecture that when  $p(j)$ are replaced by their estimates
\begin{equation}\label{phat}
\hat p(j)=\frac{1}{2M}\sum_{ i=1}^M C_{i,j}
\end{equation}
 then the largest singular value  of the resulting matrix  $\hat{\mC}$ still has the same limit   \eqref{Tracy-Widom}. But under such normalization, the entries of the matrix become dependent so this has not been worked out with mathematical rigor.
It is also natural to conjecture that when several eigenvalues   of   $\widetilde{\mC}\widetilde{\mC}'$ or of  $\hat{\mC}\hat{\mC}'$ exceed the Tracy-Widom threshold on the right hand side of \eqref{Tracy-Widom}, then the number of subpopulations is one more than the number of such eigenvalues. However, such a result is at present not available in the context of perturbed finite rank matrices as in \eqref{Eq:perturb}.

A more refined test statistic that compensates for LD  by reducing $N$, supported by simulations, is discussed in \cite{Patterson:2006}.

\section{Conclusion}
 Eigenvalues of the uncentered covariance matrix $\mC\mC'$ larger than the theoretical threshold \eqref{Threshold}, when combined with overall histogram of eigenvalues, are a consistent indicator of the presence of subpopulations in the data. We demonstrate in two proof-of-principle simulations that we are able to obtain evidence of population structure when the number of individuals is large enough. \textcolor{black}{Our theory clearly shows that the largest singular values for well-separated populations, assuming sufficient dataset size, are an order of magnitude larger than the non-significant smaller eigenvalues. Our results underscore the utility of PCA for estimating the number of subpopulations in a dataset. }

 Our estimate \textcolor{black}{of the number of subpopulations in a dataset} is conservative, and we never obtain evidence of more subpopulations than present in the simulations. As expected, we encounter a loss of power (false negatives) in small simulated data sets. \textcolor{black}{Our theoretical derivations provide a formula that describes the relationship between} accuracy in estimating the number of subpopulations, $K$, \textcolor{black}{and}  the number of individuals $M$ in the sample    and on the number $N$ of markers. \textcolor{black}{For the typical current practice, where} the number of markers often exceeds the number of individuals (i.e., $M\leq N$), \textcolor{black}{our formula shows that increasing the number of individuals, $M$, is the primary way of improving the resolution of PCA to distinguish subpopulations}.
 \textcolor{black}{Lastly, our examination of the distribution of non-significant eigenvalues indicates that departures from the assumptions of the theory, such as cryptic relatedness among individuals, affects the shape of the histogram of small eigenvalues. }

\section*{Acknowledgements}
KB gratefully acknowledges support by the National Institutes of Health under Ruth L. Kirschstein National Research Service Award \#5F32HG006411.
The research of WB  was partially supported by NSF grant \#DMS-0904720.
We thank the referees and editor for helpful comments that  improved the paper.

%%%% BBL

\end{document}